\documentclass[journal]{IEEEtran}
\usepackage{amsmath, amssymb, amsfonts, amsthm}
\usepackage{graphicx}
\usepackage[lined,boxed,commentsnumbered, ruled]{algorithm2e}
\usepackage{algpseudocode}

\newtheorem{fact}{Fact}
\newtheorem{theorem}{Theorem}
\newtheorem{definition}{Definition}

\begin{document}

\title{Irregular Fractional Repetition Code Optimization for Heterogeneous Cloud Storage}


\author{Quan Yu,
        Chi Wan Sung,~\IEEEmembership{Member,~IEEE},
        and Terence H. Chan,~\IEEEmembership{Member,~IEEE}
\thanks{Manuscript received May 15, 2013; revised September 30, 2013 and November 20, 2013. This paper was presented in part at the IEEE International Conference on Communications (ICC), Ottawa, Canada, June 2012. This work was partially supported by a grant from the University Grants Committee of the Hong Kong Special Administrative Region, China (Project No. AoE/E-02/08).}
\thanks{Quan Yu is with Department of Electronic Engineering, City University of Hong Kong, Hong Kong (e-mail: Q.Yu@my.cityu.edu.hk).}
\thanks{Chi Wan Sung is with Department of Electronic Engineering, City University of Hong Kong, Hong Kong (e-mail: albert.sung@cityu.edu.hk).}
\thanks{Terence H. Chan is with Institute for Telecommunications Research, University of South Australia, Adelaide, SA 5095, Australia (e-mail: terence.chan@unisa.edu.au).}}
\date{}

\pagestyle{empty}
\thispagestyle{empty}


%



\maketitle

\pagestyle{empty}
\thispagestyle{empty}

\begin{abstract}
This paper presents a flexible irregular model for heterogeneous cloud storage systems and investigates how the cost of repairing failed nodes can be minimized. The fractional repetition code, originally designed for minimizing repair bandwidth for homogeneous storage systems, is generalized to the irregular fractional repetition code, which is adaptable to heterogeneous environments. The code structure and the associated storage allocation can be obtained by solving an integer linear programming problem. For moderate sized networks, a heuristic algorithm is proposed and shown to be near-optimal by computer simulations.
\end{abstract}



\IEEEpeerreviewmaketitle

\section{Introduction} \label{sec:introduction}

\PARstart{C}{loud} storage is a new paradigm of storing data. It allows users to access data anywhere and anytime. Companies such as Google and Apple are providing this service through their data centers, which are network-connected. Such an architecture is called the distributed storage system (DSS). Storage nodes in a DSS are generally unreliable and subject to failure. When a failure occurs, a newcomer needs to {\em repair} the lost data by retrieving data from surviving storage nodes, called helper nodes, so as to maintain the {\em reliability} of the DSS. Besides, the DSS should be able to provide data {\em availability}, which allows users to access their data anywhere and with low delay.

To provide reliability and availability, erasure codes such as replication or Reed-Solomon (RS) code are commonly used. While replication requires less network bandwidth during node repair, RS code is more efficient in terms of storage space. In 2007, Dimakis et al. showed that there is a fundamental tradeoff between storage space and repair bandwidth~\cite{DGWR07}. Points on the tradeoff curve can be achieved by a class of codes called {\em regenerating codes}, which is based on the concept of network coding. In their formulation, a newcomer is able to recover the lost data by connecting to any $d$ surviving storage nodes, and a data collector is able to retrieve the data object by downloading data from any $k$ out of the $n$ storage nodes. We call this distributed storage model the {\em regular} model. Since then, many codes that achieve points on the tradeoff curve have been constructed (e.g.~\cite{RSKumar2011,CSKR10,SRKR10c,GRWS11}).

The design rationale of regenerating codes is to minimize repair bandwidth. These codes, however, generally incur high disk I/O access during repair, since helper nodes need to read its stored data and linearly combine them to form packets to be sent to a newcomer. The stored data that needs to be read is often much more than the data to be sent to the newcomer. The disk access bandwidth thus becomes the bottleneck. In~\cite{ITXYWJB2013}, the repair problem is considered in a different way. It aims to minimize the amount of information to be accessed when the number of node failures is smaller than the erasure correcting capability of an MDS code. Another approach is considered in~\cite{ElRouayheb10}. It proposes a new code formed by concatenating an outer MDS code with an inner fractional repetition (FR) code. We call it MDS-FR code. This code is a minimum bandwidth regenerating (MBR) code, which means that it minimizes the repair bandwidth of the system. Furthermore, it has the nice {\em uncoded repair} property: a helper node only needs to read the exact amount of data that it needs to forward to the newcomer without any processing. In other words, it minimizes both repair bandwidth and disk access bandwidth at the same time. While the original construction of the FR code in~\cite{ElRouayheb10} is based on regular graph and Steiner system, other constructions exist, which are based on bipartite graph~\cite{JJ11}, randomized algorithm~\cite{SNSK11}, resolvable designs~\cite{OOAR12}, and incidence matrix~\cite{SMKT13}. Note that the above mentioned works do not strictly follow the regular model, as they have different design considerations in mind. Another notable example is the locally repairable code~\cite{DSPAGD12,PGCHHSSY12,FOAD11}, which aims at reducing the number of nodes that need to be contacted during repair.

In this paper, we focus on heterogeneous distributed storage systems. Examples include heterogenous data centers, peer-to-peer cloud storage systems (e.g. Space Monkey)~\cite{Oceanstore}, peer-assisted cloud storage systems, and some wired or wireless caching systems~\cite{SSHKK11,NAA12}. In these applications, the storage nodes and the network links are {\em heterogeneous}, meaning that the storage capacities and costs associated with different storage nodes may not be the same, and the communication links between each pair of storage nodes may have different characteristics in terms of bandwidth, communication cost, and transmission rate. Furthermore, it is also possible that some storage nodes are not directly connected. In such an environment, new issues arise. The storage allocation problem, which focuses on how to allocate a given storage budget over the storage nodes such that the probability of successful recovery is maximized, is studied in~\cite{LDH12}. A distributed storage system in which the storage nodes have different download costs is considered in~\cite{AKG10}. In a distributed storage system with storage cost, how to allocate storage capacities among the storage nodes so as to minimize the total storage cost is investigated in~\cite{YSS11}. In~\cite{LYWL10}, the bandwidth heterogeneity is taken into account to demonstrate that the tree-structured regeneration topology is an efficient topology to reduce the regeneration time. Under functional repair model, the link costs and the impact of network topology are jointly considered in~\cite{GXS11}, and an information-theoretic study is performed in~\cite{TSCH13}.

To address the design issues of heterogenous cloud storage systems, we set up a flexible model, called the {\em irregular} model, in which the underlying network topology can be arbitrary, the storage capacities and costs of different storage nodes are allowed to be different, and the bandwidth and costs of communication links need not be the same. We relax the constraints of data repair and data retrieval in the regular model by introducing the concepts of repair overlay and retrieval sets. We use the term repair overlay to refer to the structure of an overlay network for data repairing. Note that it is called repair table in~\cite{ElRouayheb10}. In the work of~\cite{ElRouayheb10}, for single failure case, the repair overlay is restricted to be a regular graph with each vertex having degree $d$, and the graph is randomly generated. In this paper, we do not restrict the repair overlay to be a regular graph. For the general case of multiple failures, the repair overlay in~\cite{ElRouayheb10} is a Steiner system. However, the existence of a Steiner system requires the system parameters to satisfy some specific conditions, which makes the system design inflexible. In this paper, hypergraph is used to model the repair overlay, which exists for arbitrary system parameters and can be constructed easily compared with Steiner system.

Recall that the code used in~\cite{ElRouayheb10} is a concatenation of an outer MDS code and an inner FR code. We call this construction the MDS-FR code. We extend the idea and propose the use of the {\em Irregular Fractional Repetition} (IFR) code as the inner code. While it preserves the desirable uncoded repair property, it further allows more flexibility in system design. When the distributed storage system and the underlying network is heterogeneous, the IFR code can be constructed and adapted to the given environment by solving an optimization problem, thus further reducing repair bandwidth. In our formulation, we minimize the system repair cost by properly choosing the MDS-IFR code. The problem is shown to be an integer linear programming (ILP) problem. When the number of storage nodes is small, the optimal solution can be found in a reasonable time. For larger networks, we decompose the problem into subproblems and propose a heuristic solution. For small network sizes, our heuristic is shown to be nearly optimal by comparing it with the optimal ILP method.

The rest of the paper is organized as follows. A motivating example is given in Section~\ref{sec:motivation}. Section~\ref{sec:model} states our irregular model for distributed storage systems. In Section~\ref{sec:code}, we describe the construction of MDS-IFR code and its relationship with the concept of relay overlay. In Section~\ref{sec:formulation}, we formulate the repair cost minimization problem as an integer linear problem (ILP). In Section~\ref{sec:tradeoff}, we describe how the storage-repair tradeoff of our code can be found. In Section~\ref{sec:heuristic}, we design heuristic algorithms to find suboptimal repair overlay and retrieval sets. Section~\ref{sec:simulation} provides our simulation results. We conclude the paper in Section~\ref{sec:conclusion}.

\section{A Motivating Example} \label{sec:motivation}

The regular model assumes that a newcomer is able to replace a failed storage node by contacting any $d$ surviving storage nodes and a data collector can retrieve the stored data object by downloading data from any $k$ out of the $n$ storage nodes. In some practical scenarios, however, the communication costs between a newcomer and each of the surviving storage nodes are different. Furthermore, the distances and transmission rates between a data collector and each of the $n$ storage nodes vary with the location of the data collector. The $d$ surviving nodes to be contacted by a newcomer and the $k$ storage nodes to be contacted by a data collector need not be arbitrary. It is reasonable to determine some sets of helper nodes, called {\em helper sets}, for a newcomer and some subsets of the $n$ storage nodes, called {\em retrieval sets}, for a data collector. The collection of helper sets of all the $n$ storage nodes defines the repair overlay. Thus, we modify data repair and data retrieval mechanisms based on the concepts of repair overlay and retrieval sets. We only require that a newcomer can rebuild the corresponding failed node by contacting the storage nodes in any one of its helper sets and a data collector can retrieve the data object by contacting the storage nodes in any one of the retrieval sets.

Consider a distributed storage system that can tolerate single failures with the following parameters: $n=6$, and $d=k=2$. A data object consisting of four packets would be stored in this distributed storage system. For the regular model, the corresponding tradeoff between storage amount and repair bandwidth under functional repair is shown in Fig.~\ref{fi:ringtradeoff}, where the feasible region is shown as the shaded area. Note that all points on the tradeoff curve are normalized by the number of packets contained in the data object. The points below the tradeoff curve are impossible to achieve by functional repair. Clearly, they cannot be achieved by exact repair either.

\begin{figure}
  \centering
  \scalebox{0.5}{\includegraphics{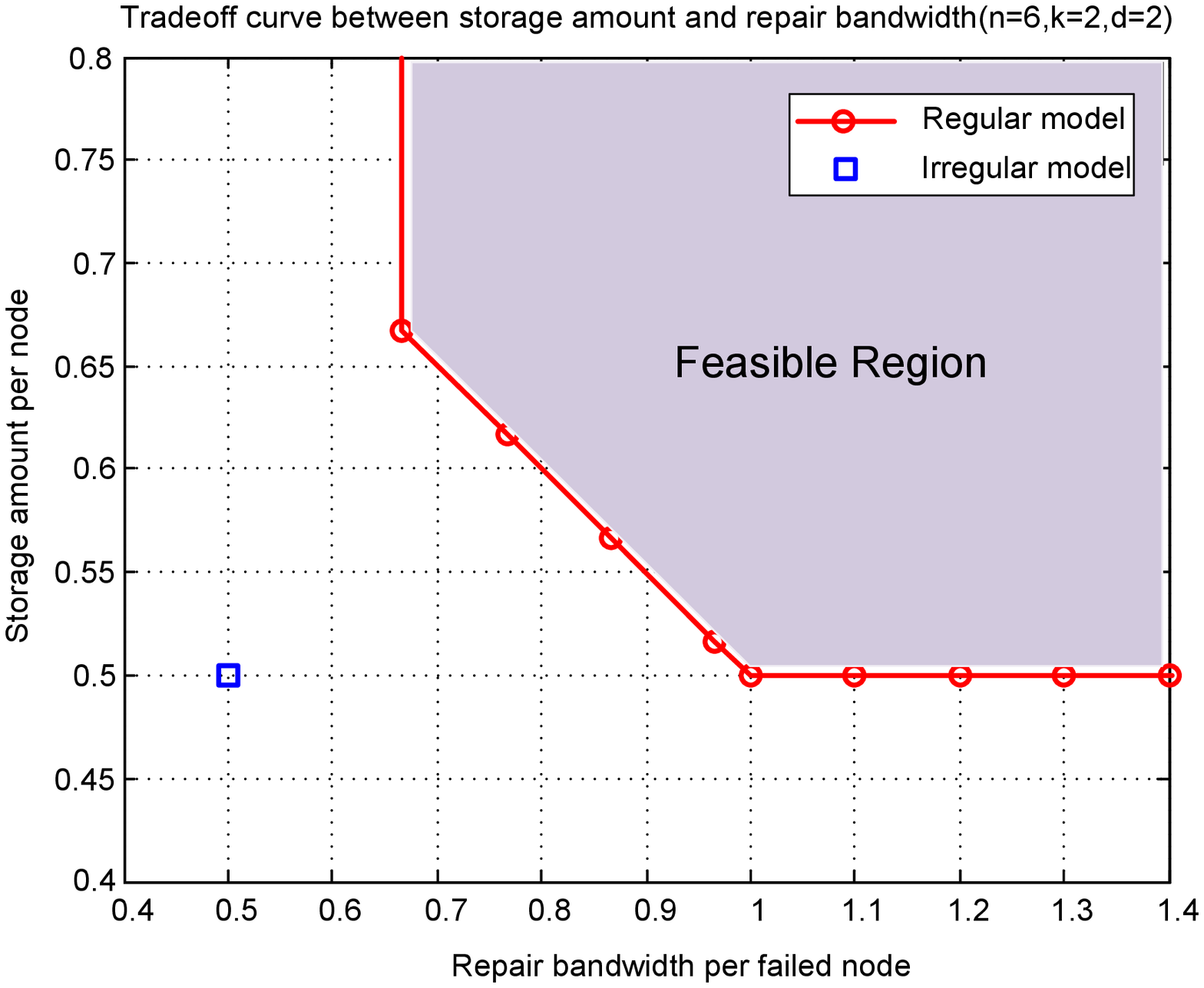}}
  \caption{Tradeoff between storage amount and repair bandwidth ($n=6, d=2$, and $k=2$).}
  \label{fi:ringtradeoff}
\end{figure}

Now we consider the irregular model, which includes the concepts of repair overlay and retrieval sets. We require that the number of retrieval sets are large enough. In this example, we require that there should be at least $9$ retrieval sets. Assume that the chosen repair overlay, denoted by $\tau$, is a ring with six nodes, as shown in Fig.~\ref{fi:ringexample}(a) (solid lines). We show how to construct MDS-FR code based on the repair overlay. The data object consisting of four packets are first encoded into six packets, $F_1$, $F_2$, \ldots, $F_6$ by a $(6,4)$-MDS code. Each edge in the 6-node ring is then associated with a coded packet. Each node stores the two packets that are associated with its incident edges, as shown in Fig.~\ref{fi:ringexample}(a). Thus, the storage amount of each of the six storage nodes is $2$. In this example, each storage node has one helper set and a newcomer can recover the lost data by connecting to $d=2$ nodes in its helper set, i.e., its two neighboring nodes in the ring, rather than {\em any} $d=2$ surviving nodes. Since each newcomer downloads one packet from each of its two helper nodes to recover the lost data, the repair bandwidth of a failed node is $2$. Suppose node 3 fails. A newcomer can replace it by downloading coded packets $F_2$ and $F_3$ from its two helper nodes 2 and 4, respectively, as shown in Fig.~\ref{fi:ringexample}(b). As for data retrieval, we can have nine retrieval sets of cardinality $k=2$, which are listed as $R_1=\{1, 3\}$, $R_2=\{1, 4\}$, $R_3=\{1, 5\}$, $R_4=\{2, 4\}$, $R_5=\{2, 5\}$, $R_6=\{2, 6\}$, $R_7=\{3, 5\}$, $R_8=\{3, 6\}$ and $R_9=\{4, 6\}$. A data collector can reconstruct the data object by connecting to the $k=2$ storage nodes in any one of the nine retrieval sets. After normalizing by the number of packets of the original data object, the storage amount of each storage node is $0.5$ and the repair bandwidth per failed node is also $0.5$. This point is also plotted in Fig.~\ref{fi:ringtradeoff}, which is below the tradeoff curve of the regular model. From Fig.~\ref{fi:ringtradeoff}, we can see that with the same storage amount per node, the repair bandwidth is reduced by $50\%$, which shows that the potential gain can be enormous {\em if the constraints of data repair and data retrieval are relaxed}.

\begin{figure}
  \centering
  \scalebox{0.45}{\includegraphics{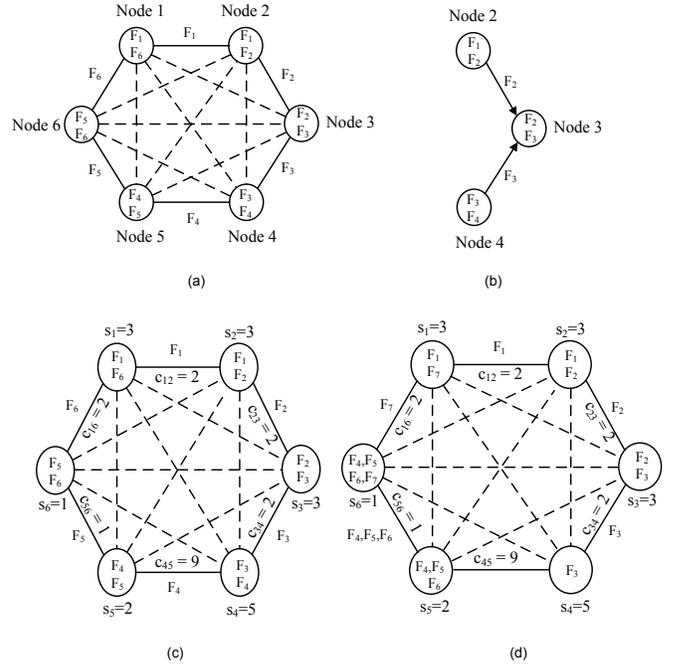}}
  \caption{An example of constructing MDS-FR code and MDS-IFR code for the irregular model.}
  \label{fi:ringexample}
\end{figure}

In the irregular model, different storage nodes are allowed to have different storage costs and different links are allowed to have different communication costs. An irregular model with storage cost and communication cost is given in Fig.~\ref{fi:ringexample}(c), where node $i$ has a (per-packet) storage cost $s_i$ and the link connecting node $i$ and node $j$ has a (per-packet) communication cost $c_{ij}$. For both FR code and IFR code, we assume that the number of packets assigned to edge $\{i, j\} \in \tau$ is $\beta_{ij}$ and the number of packets stored in node $i$ is $\alpha_i$. The total storage cost can then be obtained as $\sum_{i=1}^{6}\alpha_is_i$ and the total repair cost of all possible single node failures can be calculated as $\sum_{i=1}^{6}\sum_{\{i,j\}\in \tau}c_{ij}\beta_{ij}$. If we use the same MDS-FR code as before, the total storage cost of the six storage nodes is $2\sum_{i=1}^{6}s_i=34$ since each node stores $2$ packets, and the total repair cost of all possible single node failures can be calculated as $(c_{12}+c_{16})+(c_{12}+c_{23})+(c_{23}+c_{34})+(c_{34}+c_{45})+(c_{45}+c_{56})+(c_{56}+c_{16})=36$, where the six components of the summation correspond to the repair costs of node 1 to node 6, respectively. If we use MDS-IFR code, we first encode the data object into seven packets by using a $(7,4)$-MDS code. Then we assign coded packet $F_1$ to edge $\{1, 2\}$, $F_2$ to edge $\{2, 3\}$, $F_3$ to edge $\{3, 4\}$, $F_4$, $F_5$, and $F_6$ to edge $\{5, 6\}$, and $F_7$ to edge $\{1, 6\}$. Each node then stores the packets associated with its incident edges, as shown in Fig.~\ref{fi:ringexample}(d). In this example, a newcomer can recover the lost data by connecting to only one node or to two nodes, depending on which node is failed. This contrasts with the MDS-FR code, in which a newcomer always connects to {\em exactly} $d$ nodes. For example, if node 4 fails, a newcomer can replace it by downloading packet $F_3$ from node 3. As for data retrieval, we can have ten retrieval sets of cardinality $k=2$, which are listed as $R_1=\{1, 3\}$, $R_2=\{1, 5\}$, $R_3=\{1, 6\}$, $R_4=\{2, 5\}$, $R_5=\{2, 6\}$, $R_6=\{3, 5\}$, $R_7=\{3, 6\}$, $R_8=\{4, 5\}$, $R_9=\{4, 6\}$ and $R_{10}=\{5, 6\}$. The total storage cost of the six nodes is $2s_1+2s_2+2s_3+s_4+3s_5+4s_6=33$ and the total repair cost of all possible single node failures can be calculated as $(c_{12}+c_{16})+(c_{12}+c_{23})+(c_{23}+c_{34})+c_{34}+3c_{56}+(3c_{56}+c_{16})=22$, where the six components of the summation corresponding to the repair costs of node 1 to node 6, respectively. Compared with MDS-FR code, we can see that both storage cost and repair cost can be reduced if MDS-IFR code is adopted in the irregular model. Although the retrieval sets in the two cases are different, in the latter case, more retrieval sets are provided, which is often more desirable. Should exactly the same number of retrieval sets are needed for a fairer comparison, one can simply remove one of the retrieval sets for the latter case, as that would not affect the storage and repair costs.

In the above example, we show that there can be large performance gain in designing distributed storage systems. However, the result should be interpreted with caution. We do not claim that MDS-IFR code outperforms well-known regenerating code and MDS-FR code under their respective problem settings. In fact, they are known to be optimal under their respective problem definitions. Instead, the example serves two purposes. First, it justifies the setup of the irregular model, which is more appropriate for heterogeneous cloud storage systems. Second, it explains the irregular model and the MDS-IFR code in an intuitive way, which facilitates the understanding of the next two sections, which formally define these concepts.

\section{System Model} \label{sec:model}

Consider a distributed storage network, in which $n$ storage nodes are distributed across a wide geographical area and connected by a network with a specific topology. A data object is encoded and distributed among the $n$ storage nodes. Let the data object be represented by a collection of $B$ packets, where each packet is an element drawn from a finite field GF($q$) of size $q$. Note that a packet is the minimum unit for all storage and transmission operations in a storage system.

\subsection{Storage and Communication Costs}

We model the underlying storage network as a connected weighted undirected graph $\tilde{G} = (\mathcal{V},\tilde{\mathcal{E}})$, where the storage nodes are vertices in the vertex set $\mathcal{V}$ and the communication links correspond to the edges in the edge set $\tilde{\mathcal{E}}$. Throughout this paper, we assume that $\mathcal{V} \triangleq \{1, 2, \ldots, n\}$. Each vertex $i \in \mathcal{V}$ has an associated storage cost $s_i$ indicating the cost of storing a packet in node $i$. We define the storage cost vector ${\bf s} \triangleq [s_1,s_2,...,s_n]$. Besides, each edge $ \tilde{e} = \{i,j\} \in \tilde{\mathcal{E}}$ connecting vertices $i$ and $j$ ($i\neq j$) has an associated weight $\tilde{c}_{ij}$, called single-hop cost, which represents the cost of transmitting a packet along this edge. If there is no direct communication link between two vertices, we let the corresponding single-hop cost be infinite. The cost to transmit a packet from vertex~$i$ to vertex~$j$ is called the communication cost and is denoted by $c_{ij}$. The values of $c_{ij}$'s can be obtained from $\tilde{c}_{ij}$, depending on the underlying communication assumptions. For example, if multi-hop transmissions are allowed, then $c_{ij}$ can be defined as the cost of the minimum-cost path from~$i$ to~$j$, where the cost of a path is the sum of the single-hop costs of its constituent edges. If only single-hop transmissions are allowed, then $c_{ij}$ equals $\tilde{c}_{ij}$ for all $i$ and $j$. The matrix $\tilde{\bf C}=[\tilde{c}_{ij}]$ is called the single-hop cost matrix, and the matrix ${\bf C}=[c_{ij}]$ is called the communication cost matrix. Note that both $\tilde{\bf C}$ and ${\bf C}$ are symmetric.

In this paper, multi-hop transmissions are allowed in the underlying storage network. Since the storage network is assumed to be connected, we can construct a complete weighted graph $G = (\mathcal{V}, \mathcal{E})$ on the vertex set $\mathcal{V}$, where the weight of an edge $e = \{i,j\} \in \mathcal{E}$ is equal to the cost of the minimum-cost path between vertices $i$ and $j$, say the communication cost $c_{ij}$. We call $G$ the metric closure of $\tilde{G}$. To compute the metric closure $G$ of $\tilde{G}$, we can use Johnson's algorithm~\cite[Chapter 25]{TCRC09} to find the costs of the minimum-cost paths between all pairs of vertices in $\mathcal{V}$.

\subsection{Repair and Retrieval Requirements} \label{sec:requirement}

We formally define a distributed storage system with specific repair and retrieval requirements as follows:

\begin{definition}[Distributed Storage System]
DSS($n, \rho, d, k, w$) is a distributed storage system with $n$ storage nodes which satisfies the following requirements:
\begin{enumerate}
\item (Data Repair) As long as there are no more than $\rho$ simultaneous node failures, the lost packets of any failed node can be exactly recovered from no more than $d$ surviving nodes.
\item (Data Retrieval) A collection of $w$ retrieval sets of cardinality $k$, denoted by $\Psi \triangleq \{R_1, R_2, \ldots , R_w\}$, is specified such that the data object can be obtained from any retrieval set in $\Psi$.
\end{enumerate}
\end{definition}

In realistic distributed storage systems, $\rho$ is typically a small value. For example, the 3-replication scheme where $\rho=2$ serves the Google File System (GFS) well~\cite{GFS}. On the other hand, the data repair requirement is different from the regular model in that we do not require that a failed node can be repaired by contacting {\em any} $d$ surviving nodes. As for data retrieval, we require that the storage system has $w$ retrieval sets of cardinality $k$. This encompasses the data retrieval requirement of the regular model as a special case, which corresponds to the setting of $w=\binom{n}{k}$. Although in our formulation, all the retrieval sets have the same cardinality, it does not mean that all the storage nodes in a retrieval set need to be contacted for data retrieval, since it is allowed that the data object can be retrieved from a subset of $R_j \in \Psi$.

\section{Code Construction} \label{sec:code}

\subsection{MDS-IFR Code}

Our code construction is a concatenation of an outer MDS code and an inner Irregular Fractional Repetition (IFR) code. We call it MDS-IFR code. The data object comprised of $B$ packets is first encoded into $F$ packets over GF($q$) by using an $(F, B)$-MDS code. Note that such code exists provided that $q \geq F$ (e.g.~\cite{FJMNJAS}). In practice, the decoding complexity may be high for large values of $q$. In that case, the vector linear code recently proposed in~\cite{SG13} can be used instead. This code is easy to decode as it has the property called zigzag decodability and all computations are performed over GF(2). The price to pay is some extra storage overhead. We refer interested readers to~\cite{SG13} for details.

After encoding by the outer code, the set of the $F$ coded packets, denoted by $\mathcal{F}$, is partitioned into $\theta$ {\em coded blocks}, $\mathcal{B}_1, \mathcal{B}_2, \ldots, \mathcal{B}_\theta$, where $|\mathcal{B}_i| \triangleq \beta_i \leq B$ for all $i$. Note that $\beta_i$ denotes the number of packets in $\mathcal{B}_i$, and $F=\sum_{i=1}^{\theta}\beta_i$.  We call ${\bf b} \triangleq [\beta_1, \beta_2, \ldots, \beta_\theta]$ the {\em block assignment vector}, which will be optimized in the next section. We remark that the block assignment vector $\mathbf{b}$, rather than what packets contained in $\mathcal{B}_i$ for $i=1, 2, \ldots, \theta$, will affect the solution of minimizing the system repair cost in the next section. That is why we introduce the definition of coded blocks instead of working directly on packets. Each coded block is then replicated $\rho+1$ times and stored on $\rho+1$ different storage nodes according to an IFR code, defined as follows:

\begin{definition}[Irregular Fractional Repetition Code]
An Irregular Fractional Repetition (IFR) code $\mathcal{C}$ for DSS($n, \rho, d, \cdot, \cdot$) is a collection $\mathcal{C}$ of $n$ subsets of $\Omega \triangleq \{1, 2, \ldots, \theta \}$, satisfying the requirements that each set in $\mathcal{C}$ has cardinality at most $d$ and each element of $\Omega$ belongs to exactly $\rho+1$ sets in $\mathcal{C}$.
\end{definition}

Note that IFR code generalizes FR code in that it only requires the cardinality of each set in $\mathcal{C}$ to be no more than $d$, rather than exactly $d$. That is why we call it {\em irregular}. Besides, it addresses only the repair issue, which will become clear after we introduce the concept of repair overlay, and is independent of the parameters related to data retrieval, that is, $k$ and $w$.

An IFR code can be represented by a hypergraph $\tau=(\mathcal{V},\mathcal{H}^{\tau})$, where $\mathcal{V}$ is the vertex set and $\mathcal{H}^{\tau}\triangleq \{E_1, E_2, \ldots, E_\theta\}$ is a family of $\theta$ non-empty subsets of $\mathcal{V}$, called hyperedges. A hypergraph is said to be $\zeta$-uniform if all of its hyperedges have the same size $\zeta$. The following fact is evident:

\begin{fact}
An Irregular Fractional Repetition (IFR) code $\mathcal{C}$ for DSS($n, \rho, d, k, w$) is equivalent to a $(\rho+1)$-uniform hypergraph $\tau$ with $\theta$ hyperedges and $n$ vertices, each of which has degree less than or equal to $d$.
\end{fact}

We call such kind of hypergraph $\tau$ a {\em repair overlay}, or an {\em overlay hypergraph}. The above fact follows directly from the definitions of IFR code and uniform hypergraph, which can be illustrated by the example below.

{\bf Example}: Let $\Omega=\{1, 2, 3, 4\}$ and $\mathcal{C}=\{ \{1, 2\}, \{1, 3\}, \{1, 4\}, \{3, 4\}, \{2, 3, 4\}, \{2\}\}$. Note that $n = |\mathcal{C}| = 6$. Furthermore, it can be checked that each element of $\Omega$ belongs to three sets in $\mathcal{C}$, so $\rho = 2$. Besides, the cardinality of each set in $\mathcal{C}$ is at most three, so $d = 3$. Therefore, $\mathcal{C}$ is an IFR code for DSS$(6, 2, 3, k, w)$.

This IFR code $\mathcal{C}$ can be represented by a $3$-uniform hypergraph $\tau=(\mathcal{V}, \mathcal{H}^{\tau})$, where $\mathcal{V}=\{v_1, v_2, v_3, v_4, v_5, v_6\}$ and $\mathcal{H}^{\tau}=\{E_1=\{v_1, v_2, v_3\}, E_2=\{v_1, v_5, v_6\}, E_3=\{v_2, v_4, v_5\}, E_4=\{v_3, v_4, v_5\}\}$, as shown in Fig.~\ref{fi:IFR_Hypergraph_Example}. The hypergraph $\tau$ has $n=6$ vertices, each of which has degree less than or equal to $d=3$. Each hyperedge in $\tau$ contains $\rho+1=3$ vertices.

\begin{figure}
  \centering
  \scalebox{0.6}{\includegraphics{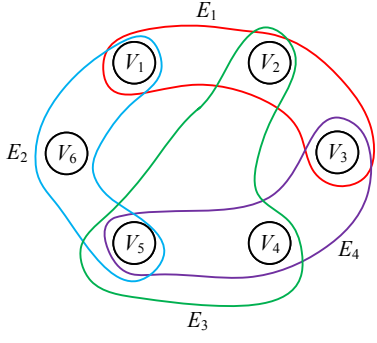}}
  \caption{A $3$-uniform hypergraph of six vertices with maximum degree $3$.}
  \label{fi:IFR_Hypergraph_Example}
\end{figure}

\subsection{Data Distribution and Data Repair}

Let $\tau = (\mathcal{V},\mathcal{H}^{\tau})$ be a given repair overlay. As described before, the data object is first encoded into $\theta$ coded blocks by an outer MDS code. For $i=1, 2, \ldots, \theta$, the coded block $\mathcal{B}_i$ is then assigned to $E_i \in \mathcal{H}^{\tau}$. All vertices contained in $E_i$ then store $\mathcal{B}_i$ in common. The storage amount $\alpha_v$ of a vertex $v \in \mathcal{V}$ can then be obtained as $\alpha_v=\sum_{i: v \in E_i} \beta_i$. Note that $F$ and $\alpha_v$ are related by $(\rho+1) F = \sum_{v \in \mathcal{V}} \alpha_v$, since each coded block is replicated $\rho+1$ times.

Data repair is very simple. When there is a node failure, a newcomer will replace the failed node by retrieving the previously stored data from a set of helper nodes. For example, suppose node $v$ which contains coded blocks $\{\mathcal{B}_i: v \in E_i\}$ fails, the newcomer can directly retrieve $\mathcal{B}_i$ from any surviving node in $E_i$ for all $i$ such that $v \in E_i$. This is what we call {\em uncoded} and {\em exact} repair. Since the cardinality of a hyperedge is $\rho+1$, this kind of repair can be done successfully provided that the number of node failures is no more than $\rho$.

\vspace{2mm}

{\bf Example}: Consider a distributed storage network $\tilde{G}$ shown in Fig.~\ref{fi:RepairExample}$(a)$. The number associated with an edge denotes the single-hop cost between its two endpoints. Fig.~\ref{fi:RepairExample}$(b)$ is the metric closure, $G$, of $\tilde{G}$, where the number associated with an edge is the corresponding communication cost. Suppose this storage network can tolerate up to $\rho=2$ node failures, and each failed node can be recovered from at most $d=3$ available storage nodes. One feasible repair overlay is shown in Fig.~\ref{fi:RepairExample}$(c)$, where every hyperedge has $\rho+1=3$ nodes and the degree of each node is less than or equal to $d=3$. Assign coded blocks $\mathcal{B}_1$, $\mathcal{B}_2$, $\mathcal{B}_3$ and $\mathcal{B}_4$ to hyperedges $E_1$, $E_2$, $E_3$ and $E_4$ respectively. Then node $1$ would store blocks $\mathcal{B}_1$ and $\mathcal{B}_2$, node $2$ would store $\mathcal{B}_1$ and $\mathcal{B}_4$, node $3$ would store $\mathcal{B}_1$, $\mathcal{B}_3$ and $\mathcal{B}_4$, node $4$ would store $\mathcal{B}_2$, $\mathcal{B}_3$ and $\mathcal{B}_4$, and node $5$ would store $\mathcal{B}_2$ and $\mathcal{B}_3$. Suppose nodes $1$ and $2$ fail. The newcomer for node~1 can download $\mathcal{B}_1$ from node~3 and download $\mathcal{B}_2$ from either node~4 or node~5, while the newcomer for node~2 can download $\mathcal{B}_1$ from node~3 and download $\mathcal{B}_4$ from node~3 or node~4.

\begin{figure}
  \centering
  \scalebox{0.6}{\includegraphics{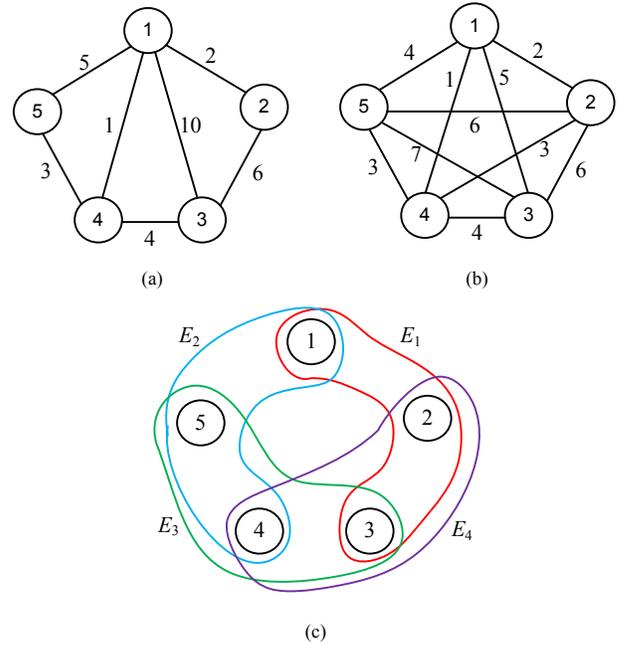}}
  \caption{An example of constructing MDS-IFR code on hypergraph.}
  \label{fi:RepairExample}
\end{figure}

\section{Repair Cost Minimization} \label{sec:formulation}

In this section, we consider the problem of minimizing repair cost. We first present an algorithm to find optimal repair order when there is more than one failed node. Based on this result, we further construct an optimization framework to determine the rate of the outer MDS code, the structure of the IFR code by means of repair overlay, the storage amount of each node, and the collection of retrieval sets.

\subsection{Repair Order under Multiple Failures}

In this subsection, we assume that a repair overlay, $\tau$, is given. We use $\Upsilon$ to denote a set of failed nodes and call it a failure pattern. Let $\Upsilon_i \triangleq \Upsilon \bigcap E_i$ be the set of failed nodes in hyperedge $E_i \in \tau$ under failure pattern $\Upsilon$. To repair all the failed nodes in $\Upsilon$, we need $|\Upsilon|$ newcomers in total. All lost blocks, i.e., $\{\mathcal{B}_i: \Upsilon_i \neq \emptyset\}$, need to be regenerated in the corresponding newcomers. This can always be done provided that $|\Upsilon| \leq \rho$. In that case, we say that $\Upsilon$ is {\em repairable}.

Let us focus on one particular lost block, $\mathcal{B}_i$. Each newcomer of a failed node in $\Upsilon_i$ needs to get a copy of $\mathcal{B}_i$ from a certain helper node, which can be a surviving node or another newcomer that has already recovered the coded block $\mathcal{B}_i$. The cost for a newcomer to repair $\mathcal{B}_i$ is simply the block size, $\beta_i$, multiplied by the communication cost between the newcomer and its helper node. If there is only one node in $\Upsilon_i$, then it is clear that the only newcomer, say node $v$, should choose a helper node $u$ that minimizes the communication cost $c_{uv}$. If there are multiple nodes in $\Upsilon_i$, the repair order will affect the total repair cost. To minimize the total repair cost for $\mathcal{B}_i$, a greedy algorithm, which is stated in Algorithm~\ref{alg:repairprocess}, can be used.

\begin{algorithm}
\caption{Repair process of coded block $\mathcal{B}_i$ in hyperedge $E_i$}
\label{alg:repairprocess}
\begin{enumerate}
 \item Pick the minimum-weight edge $e = \{u, v\} \in G$, where $u \in E_i \backslash \Upsilon_i$ and $v \in \Upsilon_i$. Then the newcomer of node $v$ chooses node $u$ to be its helper node and downloads a copy of $\mathcal{B}_i$ from node $u$ along the minimum-cost path.

 \item Remove node $v$ from $\Upsilon_i$, i.e., $\Upsilon_i \gets \Upsilon_i \backslash \{v\}$, which means that the newcomer of node $v$ has already recovered $\mathcal{B}_i$ and is able to act as helper node of other newcomers that still need to recover $\mathcal{B}_i$.

 \item Repeat Steps 1) and 2) until all the newcomers in $E_i$ recover $\mathcal{B}_i$, i.e., $\Upsilon_i= \emptyset$.
\end{enumerate}
\end{algorithm}

\begin{theorem}
For any given repairable failure pattern $\Upsilon$, Algorithm~\ref{alg:repairprocess} minimizes the cost of repairing $\mathcal{B}_i$, for all $i$ such that $\Upsilon_i \neq \emptyset$.
\end{theorem}

\begin{proof}
Replace all nodes $E_i \backslash \Upsilon_i$ by a virtual node $s$. For $v \in \Upsilon_i$, let $c_{sv} \triangleq \min \{c_{uv}: u \in E_i \backslash \Upsilon_i\}$. A weighted complete graph $K$ can be constructed on the vertex set $\Upsilon_i \cup \{s\}$, where the edge weight is the corresponding communication cost. The repairing of $\mathcal{B}_i$ is equivalent to sending $\mathcal{B}_i$ from node $s$ to all nodes in $\Upsilon_i$. Therefore, the minimum cost of repairing $\mathcal{B}_i$ is equal to $\beta_i$ times the weight of the minimum spanning tree of the graph $K$. Let $e_1, e_2, \ldots, e_{|\Upsilon_i|}$ be the sequence of edges of $G$ chosen by Algorithm~\ref{alg:repairprocess}. Let $f_1, f_2, \ldots, f_{|\Upsilon_i|}$ be the sequence of edges of $K$ chosen by the well-known Prim's algorithm~\cite[Chapter 23]{TCRC09} for finding a minimum spanning tree on graph $K$ with $s$ being the initial vertex. It can be seen that the communication cost of $e_i$ is equal to the weight of $f_i$ for all $i$. Therefore, Algorithm~\ref{alg:repairprocess} is optimal.
\end{proof}

Note that the repair processes of different coded blocks are independent, and thus can be executed in parallel. In repairing $\mathcal{B}_i$ under the failure pattern $\Upsilon$, let the set of edges chosen by Algorithm~\ref{alg:repairprocess} be denoted by $T^\Upsilon_i$. Given a repair overlay $\tau$ and a block assignment vector ${\bf b}$, the total repair cost, normalized by the object size $B$, under failure pattern $\Upsilon$ is then given by
\begin{equation} \label{repair_cost_failure}
\tilde{c}_r(\tau, {\bf b}, \Upsilon) = \frac{1}{B} \sum_{i: \Upsilon_i \neq \emptyset} \sum_{\{u,v\} \in T^\Upsilon_i}c_{uv}\beta_i.
\end{equation}

\subsection{MDS-IFR Code Optimization}

Our objective is to design the MDS-IFR code so as to minimize the system repair cost. If a non-repairable failure pattern $\Upsilon$ occurs, the whole data object will be decoded by downloading data from one of the retrieval sets and then re-encoded for storage in the newcomers. We assume that the system is properly designed so that the probability of occurrence of a non-repairable failure pattern is small. Therefore, we focus on minimizing the expected repair cost per unit data, where the expectation is taken over all repairable patterns. We call it the {\em system repair cost} and denote it by $c_r(\tau, \mathbf{b})$. Given a repair overlay $\tau$ and a block assignment vector $\mathbf{b}$, it can be written as
\begin{equation}
c_r(\tau, \mathbf{b}) \triangleq \sum_{\Upsilon: \text{repairable}} p(\Upsilon) \tilde{c}_r(\tau, \mathbf{b}, \Upsilon),
\end{equation}
where $p(\Upsilon)$ be the probability of occurrence of $\Upsilon$, on the condition that the failure pattern $\Upsilon$ is repairable.

Let $\tilde{E}_1, \tilde{E}_2, \ldots, \tilde{E}_{\binom{n}{\rho+1}}$ be all the $(\rho+1)$-subsets of $\mathcal{V}$. We use a binary variable to indicate whether $\tilde{E}_i$ belongs to the hyperedge set $\mathcal{H}^{\tau}$ of a repair overlay $\tau$:
\[
 x_i \triangleq \begin{cases}
 0 & \text{ if } \tilde{E}_i \notin \mathcal{H}^{\tau}, \\
 1 & \text{ if } \tilde{E}_i \in \mathcal{H}^{\tau}.
 \end{cases}
\]
Let ${\bf x} \triangleq [x_1, x_2, ..., x_{\binom{n}{\rho+1}}]$. We call it an {\em overlay selection vector}. To ensure that the degree of each vertex $v$ in $\mathcal{V}$ is not larger than $d$, we have the following constraints
\begin{equation}
\label{eq:degree}
\sum_{i: v \in \tilde{E}_i} x_i \leq d, \; \forall v \in \mathcal{V}.
\end{equation}
Note that the binary vector ${\bf x}$ defines a repair overlay, which we denote it by $\tau({\bf x})$.

As a coded block $\mathcal{B}_i \subseteq \mathcal{F}$ is assigned to hyperedge $\tilde{E}_i$ if and only if it would be contained in the overlay hypergraph, we therefore have the constraints
\begin{equation}
\label{eq:beta_i}
0 \leq \beta_i \leq Bx_i, \; i=1, 2, \ldots, \binom{n}{\rho+1}.
\end{equation}

For each storage node, it stores all the coded blocks associated with the hyperedges containing it. Thus the storage amount of node $v$, denoted by $\alpha_v$, can be obtained as
\begin{equation}
\alpha_v=\sum_{i: v \in \tilde{E}_i} \beta_i, \; \forall v \in \mathcal{V}.
\label{eq:storageamount}
\end{equation}
Note that considering data retrieval, $\alpha_v$ need not be greater than $B$ due to the outer MDS code. To facilitate efficient uncoded repair, however, we allow $\alpha_v$ to exceed $B$.
Given a repair overlay $\tau$ and a block assignment vector $\mathbf{b}$, we assume that there is a constraint on the {\em system storage cost}, denoted by $c_s(\tau, \mathbf{b})$, which is defined as the cost of storing one unit data object in DSS($n, \rho, d, k, w$). Let $C_s$ be the maximum allowable system storage cost. Then we have
\begin{equation}
c_s(\tau(\mathbf{x}), \mathbf{b}) \triangleq \frac{1}{B} \sum_{v=1}^{n}s_v \alpha_v \leq C_s.
\label{eq:storagecost}
\end{equation}
Note that $C_s$ is a given constant, which constrains the total storage cost in the system. We do not impose any constraint on the storage amount of each storage node. If needed, that kind of constraints can be easily added, and our proposed algorithm, to be described in a later section, can still be applied without any modification.

Recall that the DSS($n, \rho, d, k, w$) needs to satisfy the data retrieval requirement. There is a collection of retrieval sets, $\Psi = \{R_1, R_2, \ldots, R_w\}$. Part or all of these sets may be pre-determined based on considerations other than storage and repair costs. For example, if a data object is mainly needed by users in a specific geographical region, it would be more convenient if one or more retrieval sets are formed by storage nodes in that region, so that the response time for a user to download that object can be shortened. To provide more flexibility in our optimization framework, we allow $w_1 \leq w$ retrieval sets be given while the remaining $w_2 = w - w_1$ retrieval sets are obtained by our optimization procedure. For the pre-determined retrieval sets, $R_1, R_2, \ldots, R_{w_1}$, we need to ensure that each of them stores at least $B$ coded packets in $\mathcal{F}$. Therefore, we have the following constraints:
\begin{equation}
\sum_{i: \tilde{E}_i \cap R_j \neq \emptyset} \beta_i \geq B, \;\; j=1, 2, \ldots, w_1.
\label{eq:retrieval_fixed}
\end{equation}
It remains to determine the other $w_2$ retrieval sets. Denote the $k$-subsets of $\mathcal{V}$, excluding the pre-determined retrieval sets, by $Q_1, Q_2, \ldots, Q_W$, where $W \triangleq \binom{n}{k}- w_1$. To indicate whether $Q_j$ is a retrieval set or not, we introduce a binary variable
\[
y_j \triangleq \begin{cases}
0 & \text{ if } Q_j \notin \Psi, \\
1 & \text{ if } Q_j \in \Psi.
\end{cases}
\]
Let ${\bf y} \triangleq [y_1, y_2, ..., y_W]$. We call it a {\em retrieval set selection vector}. To guarantee that there are $w_2$ more retrieval sets, we have
\begin{equation}
\sum_{j=1}^W y_j=w_2.
\label{eq:retrievalsetno}
\end{equation}
Note that $\mathbf{y}$ defines a collection of retrieval sets, which we denote it by $\Psi(\mathbf{y})$. Similar as before, we have
\begin{equation}
\sum_{i: \tilde{E}_i \cap Q_j \neq \emptyset} \beta_i \geq By_j, \; j=1, 2, \ldots, W.
\label{eq:retrieval}
\end{equation}

As mentioned before, our objective function is the system repair cost of DSS($n, \rho, d, k, w$). Formally, the repair cost minimization problem can be stated as follows.
\begin{equation}
\text{Minimize}~~c_r(\tau(\mathbf{x}), \mathbf{b}) \triangleq  \sum_{\Upsilon: \text{repairable}} p(\Upsilon) \tilde{c}_r(\tau({\bf x}), \mathbf{b}, \Upsilon),
\label{eq:objtivefunc}
\end{equation}
subject to
\begin{align*}
\eqref{eq:degree}-\eqref{eq:retrieval},
\end{align*}

\begin{equation}
x_i \in \{0, 1\}, \; i=1, 2, \ldots, \binom{n}{\rho+1},
\label{eq:x_bianry}
\end{equation}

\begin{equation}
y_j \in \{0, 1\}, \; j=1, 2, \ldots, W,
\label{eq:y_bianry}
\end{equation}

\begin{equation}
\beta_i \in \mathbb{N}, \; i=1, 2, \ldots, \binom{n}{\rho+1}.
\label{eq:beta_integer}
\end{equation}
The optimization is an integer linear programming (ILP) problem, where $\mathbf{x}$, $\mathbf{y}$, and $\mathbf{b}$ are the optimization variables.

\section{Repair-Storage Tradeoff} \label{sec:tradeoff}

In our formulation, it is clear that there is a tradeoff between the system storage cost, $c_s$, and the system repair cost, $c_r$. To make this relationship more explicit, we introduce the following notions. For the ease of presentation, we assume that $w_1 = 0$ for the rest of this paper.


\begin{definition}[Achievability]
 A cost pair $(c_r^*, c_s^*)$ is $B$-achievable by the MDS-IFR code if given any data object of size $B$, there exists a repair overlay $\tau$, a collection of retrieval sets $\Psi$, and a block assignment vector ${\bf b}$ such that $c_r(\tau, {\bf b}) \leq c_r^*$, $c_s(\tau, {\bf b}) \leq c_s^*$, and the data repair and retrieval requirements are satisfied.
\end{definition}

\vspace{2mm}

Note that $\beta_i$'s must be integers no more than $B$. Therefore, the achievable region enlarges when $B$ increases. It is therefore natural to define the asymptotic achievable region for arbitrarily large value of $B$:

\begin{definition}[Asymptotic achievability]
  A cost pair $(c_r^*, c_s^*)$ is asymptotically achievable by the MDS-IFR code if for any $\epsilon > 0$, there exists for sufficiently large $B$, a repair overlay $\tau$, a collection of retrieval sets $\Psi$, and a block assignment vector ${\bf b}$ such that $c_r(\tau, {\bf b}) < c_r^* + \epsilon$, $c_s(\tau, {\bf b}) < c_s^* + \epsilon$, and the data repair and retrieval requirements are satisfied.
\end{definition}


The following result shows that the asymptotically achievable cost can be obtained by relaxing the integer constraint on ${\bf b}$:
\begin{theorem}
Given any $B$ and $C_s$, let ${\bf b}^* = [\beta_i^*]$, $\mathbf{x}^*$, and $\mathbf{y}^*$ be the solution to the repair cost minimization after relaxing the integer constraint on ${\bf b}$, and $c_r^* \triangleq c_r( \tau(\mathbf{x}^*), {\bf b}^*)$ be the corresponding system repair cost. The cost pair $(c_r^*, C_s)$ is asymptotically achievable by the MDS-IFR code.
\end{theorem}

\begin{proof}
First of all, note that $c_r$ and $c_s$ are invariant to scaling $B$ and all $\beta_{i}$'s by the same amount, no matter whether  $\beta_{i}$'s are integers or not. Suppose we scale up $B$ and $\beta_{i}$'s all by $\gamma > 1$. Then we round {\em up} all $\beta_{i}$'s to the nearest integers. By~\eqref{repair_cost_failure} and~\eqref{eq:objtivefunc}, the new system repair cost is given by
\begin{align}
\tilde{c}_r &= \frac{1}{\gamma B} \sum_{\Upsilon: \text{repairable}} p(\Upsilon) \sum_{i: \Upsilon_i \neq \emptyset} \sum_{\{u,v\} \in T^\Upsilon_i}c_{uv}(\gamma \beta_i + z_i) \\
&= \frac{1}{B} \sum_{\Upsilon: \text{repairable}} p(\Upsilon) \sum_{i: \Upsilon_i \neq \emptyset} \sum_{\{u,v\} \in T^\Upsilon_i}c_{uv} \beta_i  \nonumber \\
& \;\;\;\;\; + \frac{1}{\gamma B} \sum_{\Upsilon: \text{repairable}} p(\Upsilon) \sum_{i: \Upsilon_i \neq \emptyset} \sum_{\{u,v\} \in T^\Upsilon_i}c_{uv} z_i,
\end{align}
where $0 \leq z_i < 1$. Since $\gamma$ can be arbitrarily large, the second term can always be made smaller than $\epsilon$. Similarly, the new storage cost can be proven to be smaller than $C_s + \epsilon$ for sufficiently large $\gamma$. Therefore,  $(c_r^*, C_s)$ is asymptotically achievable.
\end{proof}

To identify the optimal tradeoff between system repair cost and system storage cost, we need to introduce the following two concepts:

\vspace{2mm}

\begin{definition}[Pareto-optimality]
 A $B$-achievable cost pair $(c_r^*, c_s^*)$ is called Pareto-optimal if and only if there does not exist other $B$-achievable cost pair $(c_r, c_s)$ such that the following two conditions are satisfied:
 \begin{enumerate}
 \item $c_r \leq c_r^*$ and $c_s \leq c_s^*$, and

 \item $c_r < c_r^*$ or $c_s < c_s^*$.
 \end{enumerate}
\end{definition}

Roughly speaking, Pareto-optimal $B$-achievable cost pairs are ``on the boundary" of the set of all $B$-achievable cost pairs. In fact, to characterize the set of $B$-achievable cost pairs, it is necessary and sufficient to characterize only the so-called Pareto frontier:

\vspace{2mm}

\begin{definition}[Pareto frontier]
The Pareto frontier is the set of all Pareto-optimal $B$-achievable cost pairs.
\end{definition}

\begin{theorem} \label{th:finite_front}
The Pareto frontier is a finite set.
\end{theorem}

\begin{proof}
Note that $\beta_i$'s are non-negative integers. According to~\eqref{eq:beta_i}, they are all less than or equal to $B$.
Since the other variables, $x_i$'s and $y_i$'s, are all binary, the solution space is finite. Hence, the Pareto frontier is finite too.
\end{proof}

Since the Pareto frontier is finite, it is possible to list all of them in finite time, which can be done by Algorithm~\ref{alg:paretofront}.

\begin{algorithm}
\caption{Find the Pareto frontier}
\label{alg:paretofront}
\BlankLine

\begin{enumerate}
\item Generate a solution which minimizes $c_r$ subject to the constraints \eqref{eq:degree}-\eqref{eq:beta_i}, \eqref{eq:retrieval_fixed}-\eqref{eq:retrieval} and \eqref{eq:x_bianry}-\eqref{eq:beta_integer}. Exit if no solution is found. Otherwise, let $c_r^*$ be the optimal value for $c_r$. \\
\item Constrain $c_r$ to be equal to $c_r^*$, and generate a solution which minimize $c_s$ subject to constraints \eqref{eq:degree}-\eqref{eq:beta_i}, \eqref{eq:retrieval_fixed}-\eqref{eq:retrieval} and \eqref{eq:x_bianry}-\eqref{eq:beta_integer}. Let $c_s^*$ be the optimal value for $c_s$. Output $(c_r^*, c_s^*)$. \\
\item Replace the constraint $c_r = c_r^*$ (which was added in Step~2) by the constraint $c_s < c_s^*$.\\
\item Go to Step 1.
\end{enumerate}
\end{algorithm}

\begin{theorem}
All the cost pairs in the Pareto frontier can be listed by Algorithm~\ref{alg:paretofront} in finite time.
\end{theorem}

\begin{proof}
By Theorem~\ref{th:finite_front}, there is a finite number of Pareto-optimal $B$-achievable cost pairs. Denote them by $(a_1, b_1), (a_2, b_2), \ldots, (a_M, b_M)$, where $M$ is the cardinality of the Pareto frontier. Furthermore, let them be ordered so that $a_i < a_j$ and $b_i > b_j$ for $i < j$.

We claim that the first point output by Algorithm~\ref{alg:paretofront} is $(a_1, b_1)$. To see this, note that it first minimizes $c_r$, without any constraint on $c_s$. The result so obtained must be $c_r^* = a_1$, for otherwise $(a_1,b_1)$ would not be the first point in the Pareto frontier. Then in Step~2, it minimizes $c_s$, with the constraint $c_r = c_r^* = a_1$. The result $c_s^*$ must be less than or equal to $b_1$, since $(a_1, b_1)$ is $B$-acheivable. Moreover, $c_s^*$ cannot be strictly less than $b_1$, for otherwise $(a_1, b_1)$ is not Pareto-optimal. As a consequence, the first point $(a_1, b_1)$ is output.

By the same argument, we can see that all points output by Algorithm~\ref{alg:paretofront} must be Pareto-optimal, and thus belong to the Pareto frontier. Assume the pair $(a_i, b_i)$ has just been output. Algorithm~\ref{alg:paretofront} first minimizes $c_r$, with the constraint $c_s < b_i$. The result so obtained must be $c_r^* > a_i$, for otherwise, $(a_i, b_i)$ is not a Pareto-optimal point. It must be equal to $a_{i+1}$, for otherwise $(a_{i+1}, b_{i+1})$ is not in the Pareto frontier. Next, Algorithm~\ref{alg:paretofront} minimizes $c_s$, with the constraint $c_r = a_{i+1}$. The result will then be $c_s^* = b_{i+1}$. Therefore, the next Pareto-optimal pair $(a_{i+1}, b_{i+1})$ is output.

As a result, all points in the Pareto frontier will be output. The algorithm terminates when no more Pareto-optimal points can be found.
\end{proof}

We remark that Algorithm 2 cannot be replaced by solving a family of weighted sum minimization problems (with different weights), since not all Pareto optimal cost pairs lie on the boundary of the convex hull of all achievable cost pairs.

\section{A Heuristic Solution} \label{sec:heuristic}

Our repair cost minimization problem is a joint repair overlay, retrieval sets, and block assignment optimization problem. In theory, it can be solved by ILP. For large network size, however, ILP is too time consuming due to the fast-growing problem dimension. In this section, we present an efficient heuristic to solve the problem.

Our heuristic algorithm is divided into the following three steps:
\begin{enumerate}
\item Determine the repair overlay, $\tau$, or equivalently, the overlay selection vector $\mathbf{x}$.
\item Determine the collection of retrieval sets, $\Psi$, or equivalently, the retrieval set selection vector, $\mathbf{y}$.
\item Determine the block assignment vector, $\mathbf{b}$.
\end{enumerate}

First, we determine the repair overlay $\tau = (\mathcal{V}, \mathcal{H}^\tau)$ by a greedy approach. We examine all $\binom{n}{\rho+1}$ possible hyperedges that could be put into $\mathcal{H}^\tau$. Let them be $\tilde{E}_1, \tilde{E}_2, \ldots, \tilde{E}_{\binom{n}{\rho+1}}$. Let $G_i$ be the subgraph of the metric closure $G$ induced by the vertices in $\tilde{E}_i$. The cost of the minimum spanning tree of $G_i$ is called the {\em MST weight} of $\tilde{E}_i$. At each step, we choose an hyperedge, not previously chosen, with the smallest MST weight while obeying the degree constraint. Such an hyperedge is then added to $\mathcal{H}^\tau$. The procedure then repeats. We formally state our method as Algorithm~\ref{alg:repairoverlay}.

\begin{algorithm}
\caption{Find a repair overlay $\tau$}
\label{alg:repairoverlay}
\KwIn{$\tilde G = (\mathcal{V}, \tilde{\mathcal{E}}), d, \rho$}
\KwOut{$\tau=(\mathcal{V}, \mathcal{H}^{\tau})$}
\BlankLine

\begin{enumerate}
\item Compute the metric closure $G$ of $\tilde G$.
\item Initialize $\tau$ with $\mathcal{H}^{\tau} \gets \emptyset$, and $N^{\tau}_v \gets 0$ for all $v \in \mathcal{V}$.
(Note that $N^{\tau}_v$ represents the degree of vertex $v$ in $\tau$.)
\item Sort all the $(\rho+1)$-subsets of $\mathcal{V}$ in ascending order of MST weight and get the sequence $\tilde{E}_1, \tilde{E}_2, \ldots, \tilde{E}_{\binom{n}{\rho+1}}$.
\item {\bf for} $i=1$ to $\binom{n}{\rho+1}$ {\bf do} \\
\hspace{4mm} {\bf if} $N^{\tau}_v < d$ for all $v\in \tilde{E}_i$ {\bf then} \\
\hspace{10mm} $\mathcal{H}^{\tau} \gets \mathcal{H}^{\tau} \cup \{\tilde{E}_i\}$  \\
\hspace{10mm} $N^{\tau}_v \gets N^{\tau}_v+1$ for all $v \in \tilde{E}_i$ \\
\hspace{4mm} {\bf end} \\
{\bf end}
\item Return $\tau=(\mathcal{V}, \mathcal{H}^{\tau})$. \\
\end{enumerate}
\end{algorithm}

Let $m$ be the number of edges in the underlying storage network $\tilde{G} = (\mathcal{V},\tilde{\mathcal{E}})$. The complexity of Algorithm~\ref{alg:repairoverlay} is $O(n^2\log n + mn + {\rho}n^{\rho+1}(\rho + \log n))$, since computing the metric closure $G$ of $\tilde{G}$ by Johnson's algorithm has the complexity of $O(n^2 \log n + mn)$ , finding the MST weight of the $\binom{n}{\rho+1}$ hyperedges has the complexity of $O({\rho}^2 n^{\rho+1})$, the sorting in the second step has complexity of $O(\rho n^{\rho+1}\log n)$ and the third step has complexity of $O(n^{\rho+1})$. Since $\rho \geq 1$, the complexity of Algorithm~\ref{alg:repairoverlay} can be simplified as $O(mn + {\rho}n^{\rho+1}(\rho + \log n))$.

Next, we need to find $\Psi$, given a fixed repair overlay $\tau$ obtained in the previous step. According to the structure of the MDS-IFR code, we know that the coded blocks associated with two different hyperedges are distinct. For this reason, we require the $k$ storage nodes in a retrieval set jointly hit as many hyperedges of $\tau$ as possible. In other words, we require $\Psi$ be the collection of the first $w$ $k$-subsets of $\mathcal{V}$ that hit the maximum number of hyperedges of $\tau$. To find it, we use a recursive approach, which is formally stated as Algorithm~\ref{alg:retrievalsets}. Note that in Step~4 of Algorithm~\ref{alg:retrievalsets}, we use $u \bigoplus \Psi$, where $u$ is a vertex and $\Psi$ is a collection of vertex sets, to denote the operation of adding $u$ to each set in $\Psi$. For example, $v_1 \oplus \{ \{v_2, v_4\}, \{v_3, v_6\}\} = \{\{v_1, v_2, v_4\}, \{v_1, v_3, v_6\}\}$.

We remark that Algorithm~\ref{alg:retrievalsets} is for the case where $w_1=0$. If there are $w_1>0$ pre-determined retrieval sets, Algorithm~\ref{alg:retrievalsets} can be applied with a very minor modification. Before a $k$-subset of $\mathcal{V}$ is added to the collection of retrieval sets $\Psi$, the algorithm first check whether that $k$-subset happens to be one of the pre-determined retrieval sets. It would be added if and only if it is not one of them. The procedure repeats until $w-w_1$ retrieval sets have been added to $\Psi$.

\begin{algorithm}
\caption{Find a collection of retrieval sets $\Psi$=RS$(\mathcal{V}, \mathcal{H}^{\tau}, k, w)$}
\label{alg:retrievalsets}
\KwIn{$\tau = (\mathcal{V}, \mathcal{H}^{\tau}), k, w$}
\KwOut{$\Psi$}
\BlankLine

\begin{enumerate}
\item {\bf if } $(\mathcal{V}=\emptyset) \bigvee (|\Psi|=w)$ {\bf do} \\
\hspace{4mm} return $\Psi$ \\
{\bf end}
\item Find $u \in \mathcal{V}$ that hits the maximum number of hyperedges in $\mathcal{H}^{\tau}$. \\
\item Let ${\tau'}=(\mathcal{V}', \mathcal{H}^{\tau'})$ be the hypergraph obtained from ${\tau}=(\mathcal{V}, \mathcal{H}^{\tau})$ by removing $u$ from $\mathcal{V}$ and all hyperedges containing $u$ from $\mathcal{H}^{\tau}$. \\
\item $\Psi \gets u \bigoplus \text{RS}(\mathcal{V}', \mathcal{H}^{\tau'}, k-1, w)$ \\
\item {\bf if} $|\Psi|<w$ {\bf do} \\
\hspace{4mm} $\Psi \gets \Psi \bigcup \text{RS}(\mathcal{V}', \mathcal{H}^{\tau'}, k, w-|\Psi|)$ \\
{\bf end}
\item Return $\Psi$
\end{enumerate}
\end{algorithm}

The complexity of Step 2 in Algorithm~\ref{alg:retrievalsets} is $O(n|\mathcal{H}^{\tau}|)$, which is the same as $O(\frac{n^2 d}{\rho+1})$, since $|\mathcal{H}^{\tau}| \leq \frac{nd}{\rho+1}$.
Step 2 would be implemented $k$ times to find a retrieval set containing $k$ vertices, and Algorithm~\ref{alg:retrievalsets} needs to find $w$ retrieval sets. Thus, the complexity of Algorithm~\ref{alg:retrievalsets} is $O(\frac{wkn^2 d}{\rho+1})$.

Last, we need to find $\mathbf{b}$, given a fixed repair overlay $\tau$ and a fixed collection of retrieval sets $\Psi$. This can be done by solving the ILP problem while fixing $\mathbf{x}$ and $\mathbf{y}$ to the values corresponding to $\tau$ and $\Psi$, respectively. Alternatively, we can solve the LP problem by relaxing the integer constraint on ${\bf b}$ if we want to minimize the asymptotically achievable cost.

For our ILP formulation of the repair cost minimization problem, the number of variables is $2\binom{n}{\rho+1} + \binom{n}{k}$ and the number of constraints is $\binom{n}{\rho+1} + n + w + 2$. If the repair overlay and the collection of retrieval sets are fixed, the number of variables can be reduced to $|\mathcal{H}^{\tau}|$ while the number of constraints can be reduced to $|\mathcal{H}^{\tau}| + w + 1$. Since $|\mathcal{H}^{\tau}| \leq nd/ (\rho+1)$ and $d \leq n-\rho+1$, the number of variables and constraints of the ILP problem can be reduced from $O(n^{\rho+1} + n^k)$ and $O(n^{\rho+1})$, both to $O(n^2)$.

To conclude, the heuristic method consists of three steps. The first two steps have complexities $O(mn + {\rho}n^{\rho+1}(\rho + \log n))$ and $O(\frac{wkn^2d}{\rho+1})$, respectively. For practical scenarios, $\rho$ is a small constant, typically equal to 1 or 2.
In theory, linear programming can be solved in polynomial time. Therefore, regarding $\rho$ as a constant, the overall computational complexity of the heuristic method is polynomial in~$n$.

\vspace{5mm}

{\bf Example}: Consider a 5-node ring, $\tilde{G}$, shown in Fig.~\ref{fi:HeurExample}$(a)$. The number associated with an edge denotes the single-hop cost between its two endpoints. Fig.~\ref{fi:HeurExample}$(b)$ shows $G$, the metric closure of $\tilde{G}$, where the number associated with an edge is the corresponding communication cost. Suppose the storage network is able to tolerate double failures, i.e., $\rho=2$, and the degree constraint is $d=3$. We need to consider all hyperedges whose cardinality is equal to $\rho+1 = 3$, i.e., $\tilde E_1=\{1, 2, 3\}, \tilde E_2=\{3, 4, 5\}, \tilde E_3=\{1, 2, 5\}, \tilde E_4=\{2, 3, 4\}, \tilde E_5=\{1, 2, 4\}, \tilde E_6=\{1, 3, 4\}, \tilde E_7=\{1, 4, 5\}, \tilde E_8=\{2, 3, 5\}, \tilde E_9=\{2, 4, 5\}, \tilde E_{10}=\{1, 3, 5\}$. Their MST weights are 5, 5, 6, 6, 7, 7, 8, 9, 9, 10, respectively. According to Algorithm~\ref{alg:repairoverlay}, the hyperedges $\tilde E_1, \tilde E_2, \tilde E_3, \tilde E_4$ and $\tilde E_7$ are successively added into $\mathcal{H}^{\tau}$. The resulting repair overlay $\tau$ is shown in Fig.~\ref{fi:HeurExample}$(c)$. Furthermore, suppose that $k=3$ and $w=6$. According to Algorithm~\ref{alg:retrievalsets}, we can obtain a collection of retrieval sets $\Psi=\{R_1=\{1,3,2\}, R_2=\{1,3,4\}, R_3=\{1,3,5\}, R_4=\{1,4,2\}, R_5=\{1,4,5\}, R_6=\{1,2,5\}\}$.

\begin{figure}
  \centering
  \scalebox{0.6}{\includegraphics{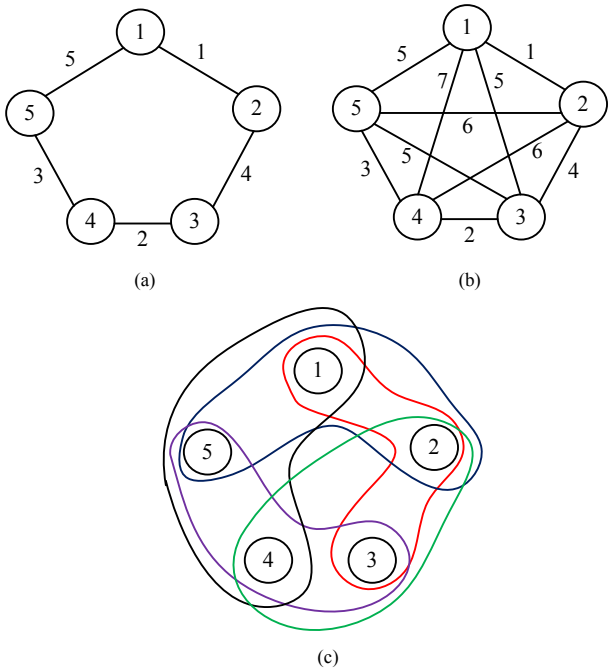}}
  \caption{An example of finding a repair overlay $\tau$ and a collection of retrieval sets $\Psi$ in a given graph $\tilde G$.}
  \label{fi:HeurExample}
\end{figure}

\section{Simulation Results} \label{sec:simulation}
In this section, we consider heterogeneous storage systems. We compare the optimal tradeoff between system storage cost and system repair cost that can be achieved by the MDS-IFR code with that achieved by the regenerating code. Moreover, we compare the minimum system repair cost that can be achieved by the MDS-IFR code with that achieved by the regenerating code for different network size. Here, we use the term ``regenerating code'' to refer to any code that achieve points on the tradeoff curve under the regular model.

In our simulations, both the storage cost vector ${\bf s}=[s_i]$ and the single-hop cost matrix $\tilde{\bf C}=[\tilde{c}_{ij}]$ are randomly generated. For the storage system whose size is less than or equal to $20$, each entry in ${\bf s}$ and $\tilde{\bf C}$ is an integer selected from the uniform distribution on the interval $[0, 50]$. For the storage system whose size is larger than $20$, each entry in ${\bf s}$ and $\tilde{\bf C}$ is an integer selected from the uniform distribution on the interval $[0, 100]$. We assume that the probabilities of occurrence of all repairable failure patterns are the same.

Consider a distributed storage system with parameters: $n=10$, $\rho=2$, $d=3$, $k=3$. For the MDS-IFR code, all Pareto-optimal $B$-achievable cost pairs $(c_s^*, c_r^*)$ can be obtained by running Algorithm~\ref{alg:paretofront} and solving the corresponding ILP problems. The curve connecting all Pareto-optimal $B$-achievable cost pairs is the optimal tradeoff between system storage cost and system repair cost that can be achieved by the MDS-IFR code, as shown in Fig.~\ref{fi:pareto}. For the regenerating code, there exists a fundamental tradeoff between the storage amount per node, $\alpha$, and the amount of data downloaded from each surviving node when repairing a failed node, $\beta$. Based on the tradeoff between $\alpha$ and $\beta$, if each newcomer downloads data along the $d$ paths with the least communication costs, the optimal tradeoff between system storage cost and system repair cost can be obtained. From Fig.~\ref{fi:pareto}, it can be seen that, compared with the regenerating code, under the same data retrieval requirement, i.e., $w=\binom{n}{k}=\binom{10}{3}=120$, the system repair cost that achieved by the MDS-IFR code can be reduced if the system storage cost are increased. However, if the data retrieval requirement are properly relaxed, i.e., $w=10$, both the system repair cost and system storage cost achieved by the MDS-IFR code can be reduced.

For the heuristic of minimizing repair cost in the irregular model by using the MDS-IFR code, to illustrate the integrality gap, we consider a distributed storage system with parameters: $n=10, d=6, k=4$, $w=\binom{10}{4}$, and $\rho=1$. We increase the data object size $B$ from $10$ to $30$, with step size $10$. For each value of $B$, we increase the maximum system storage cost per unit data object, $C_s$, from $80$ to $100$ and solve the corresponding ILP. The tradeoff curves between system storage cost, $c_s$, and system repair cost, $c_r$, are plotted in Fig.~\ref{fi:IPvsLP}. We can observe that the gap between the solution of the ILP and of its relaxation is tiny and decreases with the growing of the data object size $B$. Thus, to improve the efficiency of simulation, we solve the LP problem by relaxing the integer constraints on ${\bf b}$ to minimize the asymptotically achievable repair cost in our simulation.

We next compare the minimum system repair cost of the MDS-IFR code with that of the regenerating code for different network size. The maximum allowable system storage cost per unit data object $C_s$ is set to a sufficiently large value, $1000000$. The simulation for each value of $n$ is averaged over $100$ runs. For small storage networks, from Fig.~\ref{fi:cr_vs_n}, it can be seen that if the number of retrieval sets is equal to $\binom{n}{k}$, the minimum system repair cost that can be achieved by the MDS-IFR code is roughly reduced at least by $20 \%$. Moreover, the asymptotically achievable minimum system repair cost found by our heuristic is near-optimal. The gap between heuristic solution and optimal solution is at most $6 \%$. If the data retrieval requirement is relaxed, for example the number of retrieval sets is reduced to $w=50$, the asymptotically achievable minimum system repair cost achieved by the MDS-IFR code can be reduced at least by $70 \%$. Since the constraints of data repair and data retrieval are relaxed in the irregular model, it is not surprising that there exists a performance gain. Nevertheless, it demonstrates that there is a large room for improvement if the regular model is refined. This is particularly relevant when the networking environment is heterogeneous.

To gain more understanding about the computational efficiency of our heuristic, we increase the network size and measure its running time. The machine employed for simulation is a Dell computer with an Intel(R) Core(TM)2 Quad CPU running at 3 GHz with 4 GB RAM. The operating system is Windows 7, and the computer is a 32-bit machine. The simulation programs were written in MATLAB. Our method requires solving LP and ILP problems. These tasks were done by a free linear integer programming solver called ``lp\underline{~}solve'', which was called from our MATLAB program. In our simulation, the system parameters are set as follows: $d=5, k=4, \rho=2, w=100, Cs=1000000$ and $B=50$. The simulation for each value of $n$ is averaged over $100$ runs. The minimum system repair cost obtained by using our heuristic for different network size is shown in Fig.~\ref{fi:cr_vs_n_large}. The average running time of the three steps of our heuristic for a given problem instance is also recorded in Fig.~\ref{fi:time}. From Fig.~\ref{fi:time}, it can be seen that the most time consuming step of our heuristic is solving the LP problem after determining the repair overlay and retrieval sets. Moreover, the total time consumed by our heuristic is less than $4$ minutes when the network size is less than $150$.

\begin{figure}
  \centering
  \includegraphics[width=9cm,height=7cm]{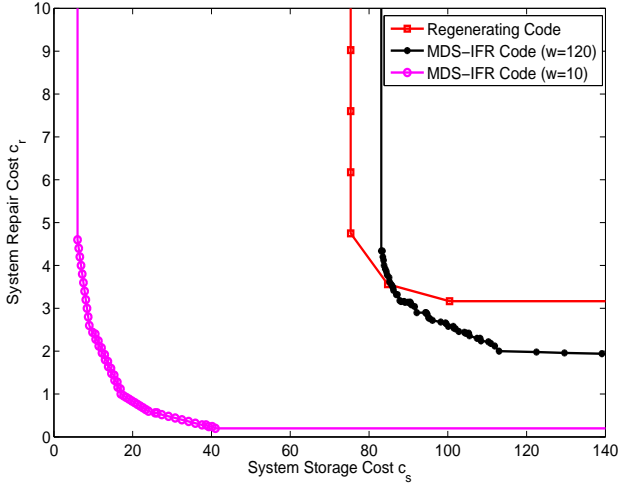}
  \caption{Optimal tradeoff between system storage cost and system repair cost ($n=10, d=4, k=3, \rho=2, \text{and}~B=20$).}
  \label{fi:pareto}
\end{figure}

\begin{figure}
  \centering
  \includegraphics[width=9cm,height=7cm]{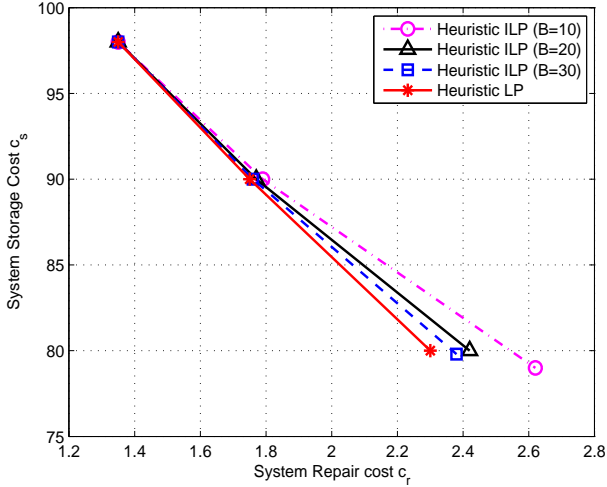}
  \caption{Tradeoff curves between system storage cost and system repair cost ($n=10, d=6, k=4, w=\binom{10}{4}=210, \text{and}~\rho=1$).}
  \label{fi:IPvsLP}
\end{figure}

\begin{figure}
  \centering
  \includegraphics[width=9cm,height=6cm]{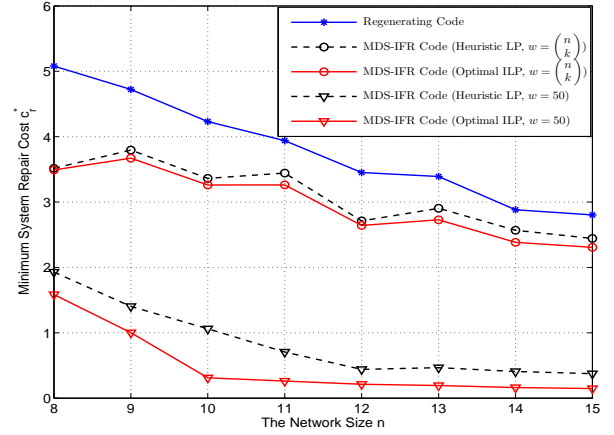}
  \caption{The minimum system repair cost for different network size ($d=4, k=3, \rho=2, \text{and}~B=30$).}
  \label{fi:cr_vs_n}
\end{figure}

\begin{figure}
  \centering
  \includegraphics[width=9cm,height=7cm]{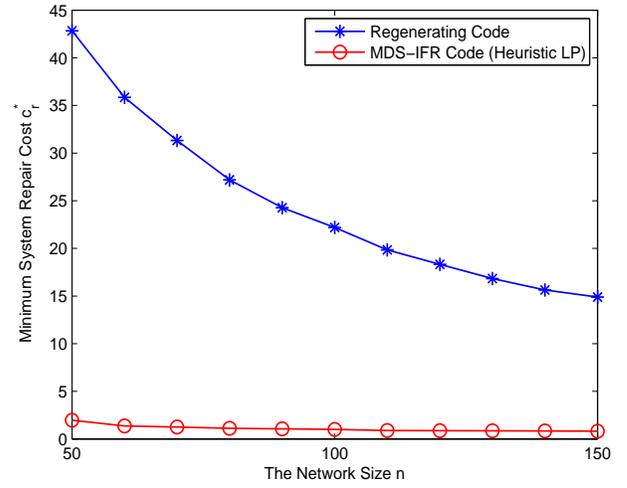}
  \caption{The minimum system repair cost for different network size ($d=5, k=4, \rho=2, w=100, \text{and}~B=50$).}
  \label{fi:cr_vs_n_large}
\end{figure}

\begin{figure}
  \centering
  \includegraphics[width=9cm,height=7cm]{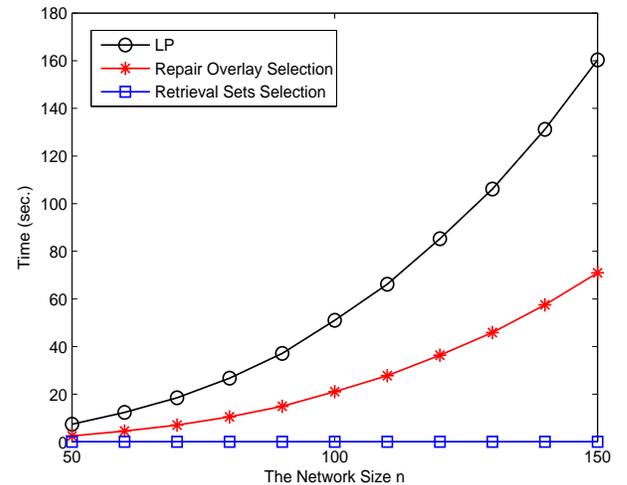}
  \caption{The average running time of the three steps of heuristic for different network size ($d=5, k=4, \rho=2, w=100, \text{and}~B=50$).}
  \label{fi:time}
\end{figure}

\newpage
\section{Conclusion} \label{sec:conclusion}

Due to the emergence of heterogeneous cloud storage systems, we generalize the concept of the FR code and propose the IFR code. A key property of the FR code is its uncoded repair process. This simple repair mechanism minimizes the repair bandwidth and the disk access bandwidth simultaneously, without any computational cost. The IFR code preserves this nice property. Moreover, its irregular structure allows the repair pattern and the storage amount of each node to be different, thus enabling the cloud system to be optimized according to network heterogeneity including different storage costs of the storage nodes and different communication costs of the links. To determine the repair pattern, which we call the repair overlay, and the storage allocation, we formulate the whole problem based on a new irregular model, with the aim of minimizing the system repair cost by properly designing the MDS-IFR code and the retrieval sets. For large networks, we decompose the repair cost minimization problem into three subproblems: repair overlay selection, retrieval sets selection, and block assignment, and propose a heuristic solution. For small network sizes, it is shown to be nearly optimal by comparing it with the optimal ILP method.

While the optimization framework established in this paper concerns mainly on system repair cost, it can be modified to include other system objectives and extended by incorporating more resource constraints. On the other hand, as it is based on the MDS-IFR code, it provides very low repair cost at the expense of higher storage overhead. If higher storage efficiency is needed in some applications, other codes will be needed
(using at the expense of higher repair cost or computing cost). This problem is beyond the scope of this paper. Nevertheless, we have demonstrated
how optimization techniques can be used to construct good codes, providing insights and new methodology on how to design future heterogeneous cloud
storage systems.

\section*{Acknowledgement}

The authors would like to thank the anonymous reviewers for their constructive comments and suggestions to improve the quality of the paper.

\vfill


\begin{thebibliography}{1}

\bibitem{DGWR07}
A.~G. Dimakis, P.~B. Godfrey, Y.~Wu, M.~J. Wainwright, and K.~Ramchandran,
  ``Network coding for distributed storage systems,'' in \emph{Proc. {IEEE} Int.
  Conf. on Computer Commun. (INFOCOM)}, Anchorage, Alaska, May 2007, pp.
  2000--2008.

\bibitem{RSKumar2011}
K.~V. Rashmi, N.~B. Shah, and P.~V. Kumar, ``Optimal exact-regenerating codes
  for distributed storage at the {MSR} and {MBR} points via a product-matrix
  construction,'' \emph{{IEEE} Trans. Inf. Theory}, vol.~57, no.~8, pp.
  5227--5239, Aug. 2011.

\bibitem{CSKR10}
C.~Suh and K.~Ramchandran, ``Exact-repair {MDS} codes for distributed storage
  using interference alignment,'' in \emph{Proc. {IEEE} Int. Symp. Inf.
  Theory (ISIT)}, Austin, Jun. 2010, pp. 161--165.

\bibitem{SRKR10c}
N.~B. Shah, K.~V. Rashmi, P.~V. Kumar, and K.~Ramchandran, ``Explicit codes
  minimizing repair bandwidth for distributed storage,'' in \emph{IEEE
  Information Theory Workshop (ITW)}, Cairo, Jan. 2010, pp. 1--5.

\bibitem{GRWS11}
A.~G. Dimakis, K.~Ramchandran, Y.~Wu, and C.~Suh, ``A survey on network codes
  for distributed storage,'' \emph{The Proceedings of the IEEE}, vol.~99,
  no.~3, pp. 476--489, Mar. 2011.

\bibitem{ITXYWJB2013}
I.~Tamo, Z.~Wang, and J.~Bruck, ``Zigzag codes: {MDS} array codes with optimal
  rebuilding,'' \emph{{IEEE} Trans. Inf. Theory}, vol.~59, no.~3, pp.
  1597--1616, Mar. 2013.

\bibitem{ElRouayheb10}
S.~{El Rouayheb} and K.~Ramchandran, ``Fractional repetition codes for repair
  in distributed storage systems,'' in \emph{Proc. 48th Annual Allerton conference on commun.
  control and computing}, Monticello, IL, Sep. 2010, pp. 1510--1517.

\bibitem{JJ11}
J.~C. Koo and J.~T. Gill, ``Scalable constructions of fractional repetition
  codes in distributed storage systems,'' in \emph{Proc. 49th Annual Allerton conference on
  commun. control and computing}, Monticello, IL, Sep. 2011, pp. 1366--1373.

\bibitem{SNSK11}
S.~Pawar, N.~Noorshams, S.~{El Rouayheb}, and K.~Ramchandran, ``{DRESS} codes for
  the storage cloud: Simple randomized constructions,'' in \emph{Proc. {IEEE}
  Int. Symp. Inf. Theory (ISIT)}, Saint Petersburg, 2011, pp. 2338--2342.

\bibitem{OOAR12}
O.~Olmez and A.~Ramamoorthy, ``Repairable replication-based storage systems
  using resolvable designs,'' in \emph{Proc. 50th Annual Allerton conference on commun. control
  and computing}, Monticello, IL, Oct. 2012, pp. 1174--1181.

\bibitem{SMKT13}
S.~Anil, M.~K. Gupta, and T.~A. Gulliver, ``Enumerating some fractional
  repetition codes,'' ar{X}iv:1303.6801 [cs.IT], Mar. 2013.

\bibitem{DSPAGD12}
D.~S. Papailiopoulos and A.~G. Dimakis, ``Locally repairable codes,'' in
  \emph{Proc. {IEEE} Int. Symp. Inf. Theory (ISIT)}, Cambridge, MA, Jul. 2012, pp.
  2771--2775.

\bibitem{PGCHHSSY12}
P.~Gopalan, C.~Huang, H.~Simitci, and S.~Yekhanin, ``On the locality of
  codeword symbols,'' \emph{{IEEE} Trans. Inf. Theory}, vol.~58, no.~11, pp.
  6952--6934, Nov. 2012.

\bibitem{FOAD11}
F.~Oggier and A.~Datta, ``Self-repairing homomorphic codes for distributed
  storage systems,'' in \emph{Proc. {IEEE} Int. Conf. on Computer Commun.
  (INFOCOM)}, Shanghai, China, Apr. 2011, pp. 1215--1223.

\bibitem{Oceanstore}
J.~Kubiatowicz, D.~Bindel, Y.~Chen, S.~Czerwinski, P.~Eaton, D.~Geels,
  R.~Gummadi, S.~Rhea, H.~Weatherspoon, W.~Weimer, C.~Wells, and B.~Zhao,
  ``Ocean{S}tore: an architecture for global-scale persistent storage,'' in
  \emph{Proc. 9th Int. Conf. on Architectural Support for Programming Languages
  and Operating Systems (ASPLOS)}, Cambridge, MA, Nov. 2000, pp. 190--201.

\bibitem{SSHKK11}
S.~Pawar, S.~E. Rouayheb, H.~Zhang, K.~Lee, and K.~Ramchandran, ``Codes for a
  distributed caching based video-on-demand system,'' in \emph{Proc. Asilomar
  Conference on Signals, Systems, and Computers}, Pacific Grove, CA, Nov. 2011,
  pp. 1783--1787.

\bibitem{NAA12}
N.~Golrezaei, A.~G. Dimakis, and A.~F. Molisch, ``Wireless device-to-device
  communications with distributed caching,'' in \emph{Proc. {IEEE} Int. Symp.
  Inf. Theory (ISIT)}, Cambridge, MA, Jul. 2012, pp. 2781--2785.

\bibitem{LDH12}
D.~Leong, A.~G. Dimakis, and T.~Ho, ``Distributed storage allocations,''
  \emph{{IEEE} Trans. Inf. Theory}, vol.~58, no.~7, pp. 4733--4752, Jul. 2012.

\bibitem{AKG10}
S.~Akhlaghi, A.~Kiani, and M.~R. Ghanavati, ``Cost-bandwidth tradeoff in
  distributed storage systems,'' \emph{Computer Communications}, vol.~33,
  no.~17, pp. 2105--2115, Nov. 2010.

\bibitem{YSS11}
Q.~Yu, K.~W. Shum, and C.~W. Sung, ``Minimization of storage cost in
  distributed storage systems with repair consideration,'' in \emph{Proc. IEEE
  Telecommunications Conference (GLOBECOM)}, Houston, Texas, Dec. 2011, pp.
  1--5.

\bibitem{LYWL10}
J.~Li, S.~Yang, X.~Wang, and B.~Li, ``Tree-structured data regeneration in
  distributed storage systems with regenerating codes,'' in \emph{Proc. {IEEE}
  Int. Conf. on Computer Commun. (INFOCOM)}, San Diego, Mar. 2010, pp. 1--9.

\bibitem{GXS11}
M.~Gerami, M.~Xiao, and M.~Skoglund, ``Optimal-cost repair in multi-hop
  distributed storage systems,'' in \emph{Proc. IEEE Int. Symp. Inf. Theory
  (ISIT)}, Saint Pertersburg, Jul. 2011, pp. 1437--1441.

\bibitem{TSCH13}
T.~Ernvall, S.~E. Rouayheb, C.~Hollanti, and H.~V. Poor, ``Capacity and
  security of heterogeneous distributed storage systems,'' in \emph{Proc.
  {IEEE} Int. Symp. Inf. Theory (ISIT)}, Istanbul, Jul. 2013, pp. 1247--1251.

\bibitem{TCRC09}
T.~H. Cormen, C.~E. Leiserson, R.~L. Rivest, and C.~Stein, \emph{Introduction
  to Algorithms}, 3rd~ed.\hskip 1em plus 0.5em minus 0.4em\relax Cambridge, {MA}: The
  MIT Press, 2009.

\bibitem{GFS}
S.~Ghemawat, H.~Gobioff, and S.-T. Leung, ``The {G}oogle file system,'' in
  \emph{Proc. {ACM} Symp. on Operating Systems Principles
  (SOSP)}, New York, Oct. 2003, pp. 29--43.

\bibitem{FJMNJAS}
F.~J. Macwilliams and N.~J.~A. Sloane, \emph{The theory of error-correcting
  codes}.\hskip 1em plus 0.5em minus 0.4em\relax New York: North-Holland, 1977.

\bibitem{SG13}
C.~W. Sung and X.~Gong, ``A zigzag-decodable code with the {MDS} property for
  distributed storage systems,'' in \emph{Proc. {IEEE} Int. Symp. Inf. Theory (ISIT)},
  Istanbul, Jul. 2013, pp. 341--345.
\end{thebibliography}
\end{document}